\documentclass[conference,letterpaper]{IEEEtran}
\IEEEoverridecommandlockouts

\usepackage{ifpdf}
\ifpdf
    \usepackage[pdftex]{graphicx}
    \graphicspath{{figs/}{matlab/}}   
    \usepackage[update]{epstopdf}
\else
	\usepackage{graphicx}
	\graphicspath{{figs/}{matlab/}}   
\fi

\usepackage[utf8]{inputenc} 
\usepackage[T1]{fontenc}
\usepackage{url}
\usepackage{ifthen}
\usepackage{amsfonts}
\usepackage{mathrsfs}
\usepackage{cite}
\usepackage[cmex10]{amsmath} 
\usepackage{amssymb,amsthm,bm,dsfont}

\usepackage{enumerate}  
\usepackage{color}
\usepackage{setspace}
\usepackage[hidelinks]{hyperref}
\usepackage{cleveref} 
\usepackage{cite}
\usepackage{subcaption}
\allowdisplaybreaks[2]  
\usepackage{multirow}

\usepackage{cases}
\usepackage{balance}
\newif\ifEnv

\Envtrue
\ifEnv
    
\else
    \excludecomment{Env}
\fi

\usepackage{optidef}

\newtheorem{theorem}{\textbf{Theorem}}
\newtheorem{prop}{\textbf{Proposition}}

\newtheorem{definition}{\textbf{Definition}}

\newcommand{\Fc}{\mathcal{F}}

\newcommand{\Pc}{\mathcal{P}}

\begin{document}
\title{\huge{Optimizing Leaky Private Information Retrieval Codes to Achieve ${O}(\log K)$ Leakage Ratio Exponent}} 


\author{%
  \IEEEauthorblockN{Wenyuan Zhao, Yu-Shin Huang, Chao Tian, and Alex Sprintson}
  \IEEEauthorblockA{Department of Electrical and Computer Engineering \\
                    Texas A\&M University\\
                    College Station, TX, 77845\\
                    Email: \{wyzhao, yushin1002, chao.tian, spalex\}@tamu.edu}
}

\maketitle


\begin{abstract}
We study the problem of leaky private information retrieval (L-PIR), where the amount of privacy leakage is measured by the pure differential privacy parameter, referred to as the leakage ratio exponent. Unlike the previous L-PIR scheme proposed by Samy et al., which only adjusted the probability allocation to the clean (low-cost) retrieval pattern, we optimize the probabilities assigned to all the retrieval patterns jointly. It is demonstrated that the optimal retrieval pattern probability distribution is quite sophisticated and has a layered structure: the retrieval patterns associated with the random key values of lower Hamming weights should be assigned higher probabilities. This new scheme provides a significant improvement, leading to an ${O}(\log K)$ leakage ratio exponent with fixed download cost $D$ and number of servers $N$, in contrast to the previous art that only achieves a $\Theta(K)$ exponent,  where $K$ is the number of messages. 
\end{abstract}

\section{Introduction}
Private information retrieval (PIR) \cite{Chor1995} is a cryptographic primitive that allows a user to retrieve information from multiple servers while safeguarding the user's privacy. In the canonical PIR setting, there are a total of $N$ non-colluding servers, where each server stores a complete copy of $K$ messages. When the user wishes to retrieve a message, the message's identity must remain hidden from any individual server. 
This privacy requirement necessarily incurs higher download costs than a protocol without such a requirement. The highest possible ratio between the message information and the downloaded cost is known as the PIR capacity, which was fully characterized by Sun and Jafar \cite{Sun2017}. A more concise optimal code was later introduced in \cite{tian2019capacity}, referred to as the TSC code, which features the shortest possible message length and the smallest query set. Variations and extensions of the canonical PIR problem have since been explored, including PIR with colluding servers \cite{t1,t2,zhou2022two}, PIR with limitations on storage \cite{c1,c2,c3,tian2023shannon,c5,c6,c7,c8,c9,c10,zhu2019new,vardy2023private}, PIR requiring symmetric privacy \cite{d1,d2,d3}, and PIR utilizing side information \cite{s1,s2,s3,s4,s5,s6,s7,lu2023single}. A more comprehensive review of these developments can be found in \cite{ulukus2022private}.

The privacy requirement in the canonical setting can be overly stringent. In practice, the privacy requirement can be relaxed, and a controlled amount of information about the query may be "leaked" to the servers. Samy et al.  \cite{Samy2019,Samy2021} studied the \emph{Leaky} PIR (L-PIR) problem, when the privacy leakage is measured by a non-negative constant $\epsilon$ within the differential privacy (DP) framework. The code used in \cite{Samy2021} is essentially a symmetrized version of the TSC code, with the probability of the "clean" retrieval pattern assigned a higher probability, and all other patterns assigned the same probability. Privacy leakage can also be measured in other  ways, e.g., maximal leakage (Max-L) \cite{Zhou2020a}, mutual information (MI) \cite{lin2021multi,qian2022improved,huang2024weakly}, and a  converse-induced measure \cite{chen2024capacity}.


\begin{figure}[t!]
    \centering
    \includegraphics[width=1\linewidth]{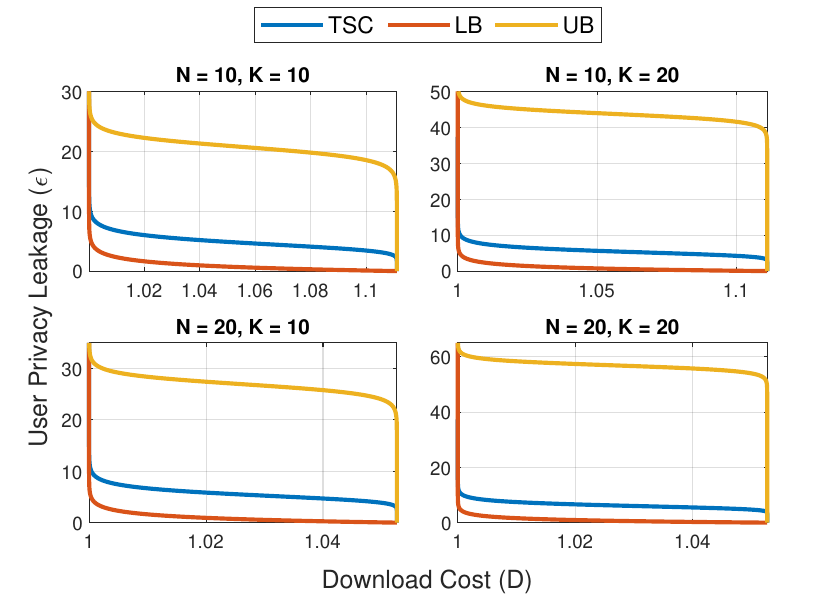}
    \caption{The effect of $(N,K)$ on the $(D, \epsilon)$ tradeoff. TSC is the proposed scheme; UB is the scheme derived in \cite{Samy2021}; LB denotes the lower bound in \cite{Samy2021}.\label{fig:DP}}
    \vspace{-0.3cm}
\end{figure}

In this work, we make the critical observation that for leaky PIR \cite{Samy2019,Samy2021},  the probability assignment based on the TSC code can be further optimized, and the optimal solution is different and considerably more sophisticated than that adopted by Samy et al. In order to find the optimal probability distribution, we first utilize a reduction that leads to a simplification of the underlying optimization problem, and then conduct a careful analysis of the duality conditions. Using the optimal probability allocation, we show that the optimized code is able to achieve $O(\log K)$ leakage ratio exponent for a fixed download cost $D$, whereas the code by Samy et al.  \cite{Samy2021} can only achieve a $\Theta(K )$ exponent. In Fig. \ref{fig:DP}, we illustrate the tradeoff between the privacy leakage and the download cost, for different $(N, K)$ values. It can be seen that the proposed scheme indeed achieves a leakage ratio exponent much lower than the prior art, which is particularly striking when $N$ and $K$ are large. 

\section{Preliminaries}

\subsection{The Canonical PIR Coding Problem}
For any positive integer $K$, we write $[1:K]:=\{1,2,\ldots,K \}$ and $W_{1:K}:=\{W_1,W_2,\ldots,W_K\}$. In a private information retrieval system, there are $N$ non-colluding servers, each of which stores a copy of $K$ independent messages $W_{1:K}$. Without loss of generality, we assume $K\geq2$, and each message consists of $L$ symbols distributed uniformly on a finite set $\mathcal{X}$. In the sequel, for any message $W:=(W[1],W[2],\ldots,W[L])$, we have
\begin{align}
    L:=H(W_1)=H(W_2)=\cdots=H(W_K),
\end{align}
where the entropy $H(\cdot)$ is calculated under the logarithm of base $|\mathcal{X}|$. A user wishes to retrieve a message $W_k$, $k\in[1:K]$, from $N$ servers. To avoid disclosing the identity of $k$ to any given server, a private random key $F^* \in \mathcal{F}^*$ is used to generate a list of (random) queries $Q^{[k]}_{1:N}$ that are sent to the corresponding servers:
\begin{align}
    Q^{[k]}_n := \phi_n(k, F^*), \quad n \in [1:N], \label{eqn:query-phi}
\end{align}
and we shall denote the union of all possible queries sent to the server-$n$ as $\mathcal{Q}_n$. 

For each $n \in [1:N]$, upon receiving a query $q \in \mathcal{Q}_n$, server-$n$ generates an answer $A^{(q)}_n$ as a function of the query $Q^{[k]}_n$ and the stored messages $W_{1:K}$:
\begin{align}
    A^{(q)}_n := \varphi_n(q, W_{1:K}), \quad n \in [1:N], \label{eqn:answer-varphi}
\end{align} 
which is represented by $\ell_n^{(q)}$ symbols in certain coding alphabet $\mathcal{Y}$. We assume $\mathcal{X}=\mathcal{Y}$ in this work to simplify the notation. In addition, we denote $A^{(Q^{[k]}_{n})}_n$ as $A^{[k]}_n$ and $\ell_n^{(Q^{[k]}_{n})}$ as $\ell_n^{[k]}$, both of which are random variables.

With all the answers received from $N$ servers, the user attempts to decode the message $\hat{W}_k$ with the function:
\begin{align}
    \hat{W}_k := \psi( A^{[k]}_{1:N}, k, F^*).
\end{align}
An information retrieval code is said to be valid only when the desired message is recovered accurately, i.e., $\hat{W}_k = W_k$.

We measure the download efficiency by the normalized (worst-case) \emph{average download cost},
\begin{align}
    D := \max_{k \in 1:K} \mathbb{E}\left[\frac{1}{L} \sum_{n = 1}^N \ell_n^{[k]} \right],
\end{align}
where $\ell_n^{[k]}$ is the length of the answer in the code and the expectation is taken with respect to the random key $F$. Note that $D$ is determined only by the queries, without being influenced by the realization of specific messages or the selection of the desired message index $k$.

\subsection{The Differential Privacy Requirement}

Under the classical perfect privacy  requirement \cite{Chor1995, Sun2017,tian2019capacity}, for each $n\in [1:N]$, the query $Q_n^{[k]}$ cannot reveal any information about the index $k$ of the message demanded to any one of the server. 
In the leaky (or the weakly private) setting, a certain amount of privacy leakage can be tolerated \cite{Zhou2020a,Samy2021,lin2021multi,qian2022improved,huang2024weakly,chen2024capacity}, which must be controlled under certain leakage measure. In this work, we adopt the measure based on differential privacy that was considered in \cite{Samy2021}.

\begin{definition}[Differential Privacy]
A randomized mechanism $\mathcal{R}: \mathcal{X} \rightarrow \mathcal{Y}$ is said to provide $(\epsilon, \xi)$-differential privacy ($(\epsilon, \xi)$-DP) if for any $x, x' \in \mathcal{X}$ that differ on a single element (i.e., a single message in the data base), 
\begin{align*}
    \mathbb{P}( \mathcal{R}(x) \in S) \leq e^\epsilon \mathbb{P}( \mathcal{R}(x') \in S) + \xi, \quad \forall S \subset \mathcal{Y}.
\end{align*}
\end{definition}


\noindent{\it $\epsilon$-differential privacy}: The $\epsilon$-differential privacy measure, i.e., ``pure differential privacy" with $\xi=0$, was used to measure the privacy leakage in \cite{Samy2021}. That is, for each $n\in[1:N]$ and $q\in \mathcal{Q}_n$, and every pair $W_{k_1}$, $W_{k_2}\in W_{1:K}$ that $k_1\neq k_2$, the following likelihood ratio needs to be bounded as follows
\begin{align}
\frac{\mathbb{P} \left( Q^{[k_1]}_{n}=q, A^{[k_1]}_{n}=a | W_{1:K} \right)}{\mathbb{P} \left( Q^{[k_2]}_{n}=q, A^{[k_2]}_{n}=a | W_{1:K} \right)} \leq e^{\epsilon}, 
\label{eq:DP-def}
\end{align}
for the privacy parameter $\epsilon \geq 0$. Clearly, the lower $\epsilon$, the more stringent the privacy requirement is. Retrieval becomes completely private when $\epsilon=0$, which is the canonical setting. We refer to the parameter $\epsilon$ as the \emph{leakage ratio exponent}, which is the main object of our study in this work. 

\section{The Base PIR Code}
The TSC code given in \cite{tian2019capacity} will play an instrumental role in this work. In this section, we review the generalized TSC code with permutation and a reduced version of the TSC code, which was used in our previous work \cite{huang2024weakly} for settings under the maximal leakage measure. Finally, we briefly introduce the main results of the L-PIR scheme in \cite{Samy2021}.

\subsection{TSC Code with Permutation}
In the TSC code, the message length is $L = N-1$. A dummy symbol $W_k[0]=0$ is prepended at the beginning of all messages. To better understand the leaky PIR system, we use a generalized TSC code with permutation in \cite{huang2024weakly}, which can be viewed as probabilistic sharing between the permutations (on $N$ servers) of the TSC code in \cite{tian2019capacity}. 

The random key $F^*$ comprises the concatenation of a random vector $F$ and a random bijective mapping $\pi: [1:N]\rightarrow [0:N-1]$ (i.e., a permutation on the set $[1:N]$ but downshifted by $1$)
\begin{align}
    F^*: = (F,\pi)=(F_1,F_2,\dots,F_{K-1},\pi),
\end{align}  
where $F$ is of length ($K-1$) and distributed uniformly in $[0:N-1]^{K-1}$. We shall use $f$ to denote a specific realization of the random vector $F$, and use $\mathcal{F}$ to denote the set of $[0:N-1]^{K-1}$. The random key $F^*$ to retrieve the message $W_k$ is generated by a probability distribution
\begin{align}
    \mathbb{P}_k(F^*)=p^{k,\pi}_{(f)}, \quad F^* = (f,\pi)\in \mathcal{F^*},
\end{align}
where $\mathcal{F}^*= ([0:N-1]^{K-1}\times \mathcal{P})$ for which $\mathcal{P}$ is the set of all bijective mappings $[1:N]\rightarrow [0:N-1]$. According to the law of total probability, $\mathbb{P}_k(F^*)$ needs to satisfy
\begin{align}
    \sum_{f\in \mathcal{F}}\sum_{\pi \in \mathcal{P}} p^{k,\pi}_{(f)}  = 1,\quad k=1,2,\ldots,K.
\end{align}

The query $Q_n^{[k]}$ to server-$n$ is generated by the function 
\begin{equation}
\begin{aligned}
     \phi_n^*(k, F^*) \triangleq (F_1, F_2, &\dots,F_{k-1}, (\pi(n)-\sum_{j = 1}^{K-1}F_j)_N, \\
     & F_{k}, F_{k+1}, \dots, F_{K-1}), \label{eqn:tsc-phi}
\end{aligned}
\end{equation}
where $(\cdot)_N$ represents the modulo $N$ operation. The corresponding answer returned by the server-$n$ is generated by
\begin{align}
    \varphi^*(q,W_{1:K})& \triangleq W_1[Q_{n,1}^{[k]}]\oplus W_2[Q_{n,2}^{[k]}]\oplus \cdots \oplus W_K[Q_{n,K}^{[k]}]\notag\\
    & = W_k[(\pi(n)-\sum_{j = 1}^{K-1}F_j)_N] \oplus \mathscr{I},
\end{align}
where $\oplus$ denotes addition in the given finite field, $Q_{n,m}^{[k]}$ represents the $m$-th symbol of $Q_{n}^{[k]}$, and $\mathscr{I}$ is the interference signal defined as
\begin{equation}
    \begin{split}
        \mathscr{I} = W_1[F_1]\oplus\cdots&\oplus W_{k-1}[F_{k-1}] \\
        &\oplus W_{k+1}[F_{k}] \oplus \cdots \oplus W_{K}[F_{K-1}].
    \end{split}
\end{equation}

The decoding procedure follows directly from the original generalized TSC code. The correctness of the code is obvious, and the download cost $D$ can be simply computed as 
\begin{align}
    p_d^k &\triangleq \sum_{\pi \in \mathcal{P}} p^{k,\pi}_{(\underline{0_{K-1}})}, k\in [1:K],\label{eq:pd}\\
    D & = \max_k \left\{ p^k_d + \frac{N}{N-1}(1-p^k_d)\right\},
\end{align}
where $\underline{0_{K-1}}$ is the length-($K-1$) all-zero vector, and $p_d^k$ is the overall probability of using a direct download by retrieving message $W_k$ from $(N-1)$ servers. 

In our previous studies \cite{qian2022improved,huang2024weakly},  a direct retrieval pattern from a single server was also adopted when the leakage was measured by maximal leakage and mutual information. However, this particular retrieval pattern immediately leads to an unbounded leakage ratio exponent under the DP measure, and we exclude it in this work.


\subsection{The Reduced TSC Code}

\begin{table*}[tb!]
\caption{The Reduced TSC code for $N=3, K=2$}
\label{tab:N3K2reduced}
    \centering
    \begin{subtable}[t]{0.495\textwidth}
        \caption{Retrieval of $W_{1}$}
        \centering
        \resizebox{\linewidth}{!}{
        \renewcommand{\arraystretch}{1.2}
        \begin{tabular}{|ccccccccc|}
        \hline
        \multicolumn{9}{|c|}{Requesting Message $k=1$}                                                                            \\ \hline
        \multicolumn{1}{|c|}{\multirow{2}{*}{Prob.}} & \multicolumn{1}{c|}{\multirow{2}{*}{$F$}} & \multicolumn{1}{c|}{\multirow{2}{*}{$\pi$}}&\multicolumn{2}{c|}{Server 1}                                         & \multicolumn{2}{c|}{Server 2}                                         & \multicolumn{2}{c|}{Server 3}                    \\ \cline{4-9} 
        \multicolumn{1}{|c|}{}                       & \multicolumn{1}{c|}{}      & \multicolumn{1}{c|}{}               & \multicolumn{1}{c|}{$Q_{1}^{[1]}$} & \multicolumn{1}{c|}{$A_{1}$}     & \multicolumn{1}{c|}{$Q_{2}^{[1]}$} & \multicolumn{1}{c|}{$A_{2}$}     & \multicolumn{1}{c|}{$Q_{3}^{[1]}$} & $A_{3}$     \\ \hline\hline
        \multicolumn{1}{|c|}{$p_0$}    & \multicolumn{1}{c|}{$0$}  & \multicolumn{1}{c|}{$(2,1,0)$}                & \multicolumn{1}{c|}{$20$}          & \multicolumn{1}{c|}{$a_2$}       & \multicolumn{1}{c|}{$10$}          & \multicolumn{1}{c|}{$a_1$}       & \multicolumn{1}{c|}{$00$}          & $\emptyset$ \\ \hline
        \multicolumn{1}{|c|}{$p_0$}    & \multicolumn{1}{c|}{$0$}    & \multicolumn{1}{c|}{$(1,0,2)$}              & \multicolumn{1}{c|}{$10$}          & \multicolumn{1}{c|}{$a_1$}       & \multicolumn{1}{c|}{$00$}          & \multicolumn{1}{c|}{$\emptyset$} & \multicolumn{1}{c|}{$20$}          & $a_2$       \\ \hline
        \multicolumn{1}{|c|}{$p_0$}    & \multicolumn{1}{c|}{$0$}    & \multicolumn{1}{c|}{$(0,2,1)$}              & \multicolumn{1}{c|}{$00$}          & \multicolumn{1}{c|}{$\emptyset$} & \multicolumn{1}{c|}{$20$}          & \multicolumn{1}{c|}{$a_2$}       & \multicolumn{1}{c|}{$10$}          & $a_1$       \\ \hline \hline
        
        \multicolumn{1}{|c|}{$p_1$}    & \multicolumn{1}{c|}{$1$}    & \multicolumn{1}{c|}{$(2,1,0)$}              & \multicolumn{1}{c|}{$11$}          & \multicolumn{1}{c|}{$a_1\oplus b_1$}   & \multicolumn{1}{c|}{$01$}          & \multicolumn{1}{c|}{$b_1$}       & \multicolumn{1}{c|}{$21$}          & $a_2\oplus b_1$   \\ \hline
        \multicolumn{1}{|c|}{$p_1$}    & \multicolumn{1}{c|}{$1$}    & \multicolumn{1}{c|}{$(1,0,2)$}              & \multicolumn{1}{c|}{$01$}          & \multicolumn{1}{c|}{$b_1$}       & \multicolumn{1}{c|}{$21$}          & \multicolumn{1}{c|}{$a_2\oplus b_1$}   & \multicolumn{1}{c|}{$11$}          & $a_1\oplus b_1$   \\ \hline
        \multicolumn{1}{|c|}{$p_1$}    & \multicolumn{1}{c|}{$1$}    & \multicolumn{1}{c|}{$(0,2,1)$}              & \multicolumn{1}{c|}{$21$}          & \multicolumn{1}{c|}{$a_2\oplus b_1$}   & \multicolumn{1}{c|}{$11$}          & \multicolumn{1}{c|}{$a_1\oplus b_1$}   & \multicolumn{1}{c|}{$01$}          & $b_1$       \\ \hline
        \multicolumn{1}{|c|}{$p_1$}    & \multicolumn{1}{c|}{$2$}    & \multicolumn{1}{c|}{$(2,1,0)$}              & \multicolumn{1}{c|}{$02$}          & \multicolumn{1}{c|}{$b_2$}       & \multicolumn{1}{c|}{$22$}          & \multicolumn{1}{c|}{$a_2\oplus b_2$}   & \multicolumn{1}{c|}{$12$}          & $a_1\oplus b_2$   \\ \hline
        \multicolumn{1}{|c|}{$p_1$}    & \multicolumn{1}{c|}{$2$}    & \multicolumn{1}{c|}{$(1,0,2)$}              & \multicolumn{1}{c|}{$22$}          & \multicolumn{1}{c|}{$a_2\oplus b_2$}   & \multicolumn{1}{c|}{$12$}          & \multicolumn{1}{c|}{$a_1\oplus b_2$}   & \multicolumn{1}{c|}{$02$}          & $b_2$       \\ \hline
        \multicolumn{1}{|c|}{$p_1$}    & \multicolumn{1}{c|}{$2$}    & \multicolumn{1}{c|}{$(0,2,1)$}              & \multicolumn{1}{c|}{$12$}          & \multicolumn{1}{c|}{$a_1\oplus b_2$}   & \multicolumn{1}{c|}{$02$}          & \multicolumn{1}{c|}{$b_2$}       & \multicolumn{1}{c|}{$22$}          & $a_2\oplus b_2$   \\ \hline
        \end{tabular}
        }
    \end{subtable}
    \begin{subtable}[t]{0.495\textwidth}
        \caption{Retrieval of $W_{2}$}
        \centering
        \resizebox{\linewidth}{!}{
        \renewcommand{\arraystretch}{1.18}
        \begin{tabular}{|ccccccccc|}
        \hline
        \multicolumn{9}{|c|}{Requesting Message $k=2$}                                                                           \\ \hline
        \multicolumn{1}{|c|}{\multirow{2}{*}{Prob.}} & \multicolumn{1}{c|}{\multirow{2}{*}{$F$}} & \multicolumn{1}{c|}{\multirow{2}{*}{$\pi$}}&\multicolumn{2}{c|}{Server 1}                                         & \multicolumn{2}{c|}{Server 2}                                         & \multicolumn{2}{c|}{Server 3}                    \\ \cline{4-9} 
        \multicolumn{1}{|c|}{}                       & \multicolumn{1}{c|}{}      & \multicolumn{1}{c|}{}               & \multicolumn{1}{c|}{$Q_{1}^{[2]}$} & \multicolumn{1}{c|}{$A_{1}$}     & \multicolumn{1}{c|}{$Q_{2}^{[2]}$} & \multicolumn{1}{c|}{$A_{2}$}     & \multicolumn{1}{c|}{$Q_{3}^{[2]}$} & $A_{3}$     \\ \hline\hline
        \multicolumn{1}{|c|}{$p_0$}    & \multicolumn{1}{c|}{$0$}  & \multicolumn{1}{c|}{$(2,1,0)$}                & \multicolumn{1}{c|}{$02$}          & \multicolumn{1}{c|}{$b_2$}       & \multicolumn{1}{c|}{$01$}          & \multicolumn{1}{c|}{$b_1$}       & \multicolumn{1}{c|}{$00$}          & $\emptyset$ \\ \hline
        \multicolumn{1}{|c|}{$p_0$}    & \multicolumn{1}{c|}{$0$}    & \multicolumn{1}{c|}{$(1,0,2)$}              & \multicolumn{1}{c|}{$01$}          & \multicolumn{1}{c|}{$b_1$}       & \multicolumn{1}{c|}{$00$}          & \multicolumn{1}{c|}{$\emptyset$} & \multicolumn{1}{c|}{$02$}          & $b_2$       \\ \hline
        \multicolumn{1}{|c|}{$p_0$}    & \multicolumn{1}{c|}{$0$}    & \multicolumn{1}{c|}{$(0,2,1)$}              & \multicolumn{1}{c|}{$00$}          & \multicolumn{1}{c|}{$\emptyset$} & \multicolumn{1}{c|}{$02$}          & \multicolumn{1}{c|}{$b_2$}       & \multicolumn{1}{c|}{$01$}          & $b_1$       \\ \hline \hline
        
        \multicolumn{1}{|c|}{$p_1$}    & \multicolumn{1}{c|}{$1$}    & \multicolumn{1}{c|}{$(2,1,0)$}              & \multicolumn{1}{c|}{$11$}          & \multicolumn{1}{c|}{$a_1\oplus b_1$}   & \multicolumn{1}{c|}{$10$}          & \multicolumn{1}{c|}{$a_1$}       & \multicolumn{1}{c|}{$12$}          & $a_1\oplus b_2$   \\ \hline
        \multicolumn{1}{|c|}{$p_1$}    & \multicolumn{1}{c|}{$1$}    & \multicolumn{1}{c|}{$(1,0,2)$}              & \multicolumn{1}{c|}{$10$}          & \multicolumn{1}{c|}{$a_1$}       & \multicolumn{1}{c|}{$12$}          & \multicolumn{1}{c|}{$a_1\oplus b_2$}   & \multicolumn{1}{c|}{$11$}          & $a_1\oplus b_1$   \\ \hline
        \multicolumn{1}{|c|}{$p_1$}    & \multicolumn{1}{c|}{$1$}    & \multicolumn{1}{c|}{$(0,2,1)$}              & \multicolumn{1}{c|}{$12$}          & \multicolumn{1}{c|}{$a_1\oplus b_2$}   & \multicolumn{1}{c|}{$11$}          & \multicolumn{1}{c|}{$a_1\oplus b_1$}   & \multicolumn{1}{c|}{$10$}          & $a_1$       \\ \hline
        \multicolumn{1}{|c|}{$p_1$}    & \multicolumn{1}{c|}{$2$}    & \multicolumn{1}{c|}{$(2,1,0)$}              & \multicolumn{1}{c|}{$20$}          & \multicolumn{1}{c|}{$a_2$}       & \multicolumn{1}{c|}{$22$}          & \multicolumn{1}{c|}{$a_2\oplus b_2$}   & \multicolumn{1}{c|}{$21$}          & $a_2\oplus b_1$   \\ \hline
        \multicolumn{1}{|c|}{$p_1$}    & \multicolumn{1}{c|}{$2$}    & \multicolumn{1}{c|}{$(1,0,2)$}              & \multicolumn{1}{c|}{$22$}          & \multicolumn{1}{c|}{$a_2\oplus b_2$}   & \multicolumn{1}{c|}{$21$}          & \multicolumn{1}{c|}{$a_2\oplus b_1$}   & \multicolumn{1}{c|}{$20$}          & $a_2$       \\ \hline
        \multicolumn{1}{|c|}{$p_1$}    & \multicolumn{1}{c|}{$2$}    & \multicolumn{1}{c|}{$(0,2,1)$}              & \multicolumn{1}{c|}{$21$}          & \multicolumn{1}{c|}{$a_2\oplus b_1$}   & \multicolumn{1}{c|}{$20$}          & \multicolumn{1}{c|}{$a_2$}       & \multicolumn{1}{c|}{$22$}          & $a_2\oplus b_2$   \\ \hline
        \end{tabular}}
    \end{subtable}
\end{table*}

A simpler scheme can in fact be as good as the general TSC code in some cases. This reduced version is later used to establish the optimal probability allocation for L-PIR, which is shown to be optimal under the $\epsilon$-DP measure given in Section \ref{sec:opt-reduced-gTSC}. In this reduced version, we set the probability as follows
\begin{align}
    \mathbb{P}_{k}(F^*) = \begin{cases} 
    p_{j}, & F^* = (F,\pi)\in \mathcal{F}_j \times \mathcal{P}_c\\
    0, & \text{otherwise}
    \end{cases},\label{eqn:reduced}
\end{align}
where $\mathcal{F}_j \triangleq \{ F\in \mathcal{F}: \|F\|=j \}$ is the set that the Hamming weight of $F$ equals to $j$ and $\mathcal{P}_c \triangleq \{ \pi \in \mathcal{P}: \pi(n+1)= \left( \pi(n)+1 \right)_N, \forall n\in [1:N] \}$ is the set that $\pi$ is \emph{cyclic}. In other words, only cyclic permutations are allowed, instead of the entire set of permutations; moreover, random vectors $F$ with the same Hamming weight are assigned the same probability. Note that this reduced TSC code is \emph{symmetric}. 


An example of the reduced code scheme (with adjusted probabilities) is given in Table \ref{tab:N3K2reduced}.

\subsection{The L-PIR Code}

The L-PIR code was proposed in \cite{Samy2019}, which is essentially a special case of the reduced TSC code discussed above, with a specific assignment of the probability allocation. More specifically, the probability of the minimum download pattern, i.e., the direct download probability $p_0$, is chosen to be
\begin{align}
    p_0=\frac{e^{\epsilon}}{Ne^{\epsilon}+N^{K}-N},
\end{align}
while the other high-cost patterns are assigned with equal probability $p_1$. The download cost achieved is
\begin{align}
    D^{\text{UB}}(\epsilon)=1+\frac{N^{K-1}-1}{\left( N-1 \right) \left( e^{\epsilon}+N^{K-1}-1 \right)}. \label{eqn:Dub}
\end{align}
We refer to the performance with this particular assignment as the performance upper bound (UB) in the sequel.

This assignment makes intuitive sense, as it increases the probability of the retrieval pattern with the lowest download cost at the expense of privacy leakage, while keeping the other retrieval patterns of equal probability since their download costs are the same. However, a key observation we make is that the patterns cannot be completely decoupled, and the optimal probability assignment is much more sophisticated. 

A lower bound (LB) was also derived in \cite{Samy2021}: 
\begin{align}
    D^{\text{LB}}(\epsilon) = 1 + \frac{1}{Ne^{\epsilon}} + \cdots + \frac{1}{(Ne^{\epsilon})^{K-1}}. \label{eqn:LB}
\end{align}
It was shown that the gap between the upper bound and the lower bound can be bounded by a (small) multiplicative constant depending on $N$ and $\epsilon$. As we shall show shortly, the gap in the leakage ratio exponents can still be quite large.

\section{Main Result}
We summarize the main result of L-PIR with $\epsilon$ -DP in the following theorems. 

\begin{theorem}
\label{thrm:epsbound}
With a fixed number of servers and a fixed feasible download cost, the leakage ratio exponent $\epsilon$'s are bounded as
\begin{align}
    \epsilon^{\text{TSC}} &\leq \log \left( K-1\right)  + \log \frac{N-1}{\alpha}, \\
    \epsilon^{\text{UB}} &\leq (K-1)\log N + \log \frac{1-\alpha}{\alpha}, \\
    \epsilon^{\text{UB}} &\geq (K-2)\log N + \log \frac{1-\alpha}{\alpha},
\end{align}
where $\alpha := (D-1)(N-1)$. 
As a direct consequence, we have $\epsilon^{\text{UB}}=\Theta\left(K\right)$, while $\epsilon^{\text{TSC}}=O(\log K)$.
\end{theorem}


The optimized allocation of TSC probability in our work is able to achieve $O(\log K)$ leakage, which is a significant order-wise improvement over $\Theta\left(K \right)$.
Theorem \ref{thrm:epsbound} is obtained by  optimizing the probability allocation of the TSC code, which is established in Theorem \ref{thrm:allocation-DP-PIR} as follows.

\begin{theorem}
\label{thrm:allocation-DP-PIR}
The optimal probability allocation for the TSC code with $\epsilon$-DP is given by the reduced scheme
\begin{align*}
    p_j = \frac{e^{(K-1-j)\epsilon }}{N\left( e^{\epsilon }+N-1\right)^{K-1}  } , ~\forall j\in [0:K-1].
\end{align*}
The corresponding optimal download cost is 
\begin{align}
    D^{\text{TSC}}\left( \epsilon \right) =1+\frac{\left( e^{\epsilon }+N-1\right)^{K-1}  -e^{(K-1)\epsilon }}{(N-1)\left( e^{\epsilon }+N-1\right)^{K-1} } . \label{eqn:optD}
\end{align}
\end{theorem}


Theorem \ref{thrm:allocation-DP-PIR} implies that without loss of optimality, we only need the cyclic permutations of the TSC code, instead of all possible permutations between servers. Unlike \emph{only} biasing the lowest download path $p_0$ in \cite{Samy2021}, we observed a layered structure in the optimal probability allocation of the TSC code with $\epsilon$-DP: the probability of retrieving the message using a random key with lower Hamming weight is assigned a higher value. More specifically, the probability ratio for using random keys from $\mathcal{F}_{j-1}$ to $\mathcal{F}_{j}$ is exactly $e^{\epsilon}$. Fig. \ref{fig:DP-d} shows the download cost with fixed leakage ratio exponents and multiplicative gaps between TSC, UB, and LB. It can be seen that the gap between UB and TSC is striking when $N$ and $K$ are large.


\begin{figure}[t!]
    \centering
    \includegraphics[width=\linewidth]{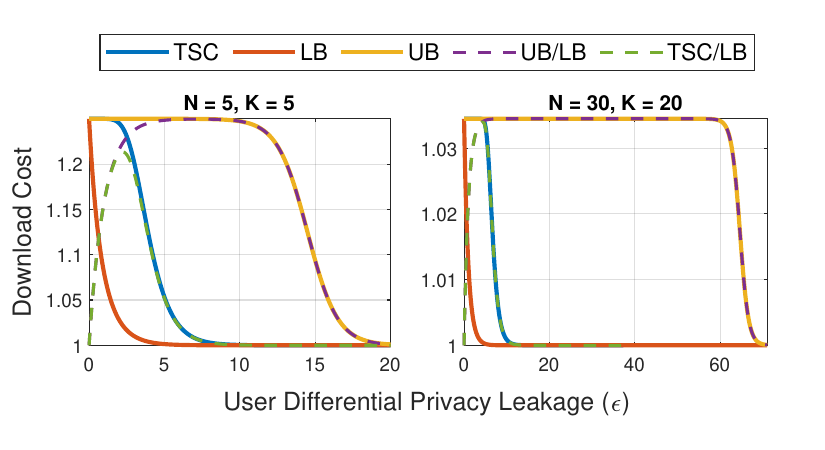}
    \caption{The download cost $D$ and multiplicative gap between TSC, UB, and LB when user privacy leakage $\epsilon$ is fixed.}
    \label{fig:DP-d}
    \vspace{-0.3cm}
\end{figure}

\section{Proof of Theorem \ref{thrm:epsbound}}
We first define $\alpha:=(D-1)(N-1)$ and assume $N, K\geq 2$. Given the LB in \eqref{eqn:LB}, we assume $D>1$ since $D=1$ if and only if there is no user privacy ($\epsilon \rightarrow \infty)$.

With Theorem \ref{thrm:allocation-DP-PIR}, we can rewrite \eqref{eqn:optD} as
\begin{align}
    \alpha &=1-\left( 1-\frac{N-1}{e^{\epsilon^{\text{TSC}} }+N-1} \right)^{K-1} \\
    &\leq 1-\left( 1-(K-1)\frac{(N-1)}{e^{\epsilon^{\text{TSC}} }+N-1} \right) \\
    &=\frac{(K-1)(N-1)}{e^{\epsilon^{\text{TSC}} }+N-1} ,
\end{align}
where the inequality follows the fact that $(1-x)^{K-1} \geq 1-(K-1)x$. 
Hence, $\epsilon^{\text{TSC}}$ can be bounded by
\begin{align}
    e^{\epsilon^{\text{TSC}} } \leq \frac{(K-1)(N-1)}{\alpha } - (N-1) \leq \frac{K-1}{D-1},
\end{align}
which implies that
\begin{align}
    \epsilon^{\text{TSC}} \leq \log \left( K-1\right)  +\log \frac{N-1}{\alpha}.
\end{align}

In a similar manner, we can rewrite \eqref{eqn:Dub} as follows:
\begin{align}
    \alpha=1-\frac{e^{\epsilon^{\text{UB} } }}{e^{\epsilon^{\text{UB} } }+N^{K-1}-1} , \quad \epsilon^{\text{UB}} \geq 0. \label{eqn:alpUB}
\end{align}
Note that $0< \alpha<1$. Hence, we have
\begin{align}
    e^{\epsilon^{\text{UB} } }=\frac{1-\alpha }{\alpha } \left( N^{K-1}-1\right),
\end{align}
which implies that
\begin{align}
    \frac{1-\alpha }{\alpha } N^{K-2} \leq e^{\epsilon^{\text{UB} } } \leq \frac{1-\alpha }{\alpha } N^{K-1}.
\end{align}
Hence, $\epsilon^{\text{UB}}$ can be bounded by
\begin{align}
    &\epsilon^{\text{UB} } \leq \left( K-1\right)  \log N+\log \frac{1-\alpha }{\alpha }, \\
    &\epsilon^{\text{UB} } \geq \left( K-2\right)  \log N+\log \frac{1-\alpha }{\alpha }.
\end{align}

\section{Proof of Theorem \ref{thrm:allocation-DP-PIR}}

The proof of Theorem \ref{thrm:allocation-DP-PIR} is decomposed into two steps: A) we first establish that we can restrict the probability allocation to the reduced TSC code without loss of optimality; B) we derive the optimal probability allocation under this reduced TSC code by carefully solving the KKT conditions.

\subsection{Optimality of the Reduced Scheme}
\label{sec:opt-reduced-gTSC}

Let's consider the L-PIR with $\epsilon$-DP, which we shall refer to as problem $(P1)$: 
\begin{equation}
\begin{aligned}
\min_{p_{(f)}^{k,\pi}} \quad & \max_{k}\left\{ D = p_d^k +\frac{N}{N-1} \left( 1-p_d^k \right)\right\} \\
\textrm{s.t.} \quad 
  & \sum_{\pi: \phi^*_n(k_1,(f,\pi))=q}p^{k_1,\pi}_{(f)} \leq \sum_{\pi :\phi^*_n(k_2,(f,\pi))=q}p^{k_2,\pi}_{(f)}e^{\epsilon},  \\
  & \qquad\qquad\qquad\qquad\qquad\qquad \forall k_1,k_2,n,q, \\
  & p_{(f)}^{k,\pi} \geq 0, \quad \forall k,\pi,f, \\
  & \sum_{f}\sum_{\pi}p^{k,\pi}_{(f)}=1, \quad \forall k.
\end{aligned}
\end{equation}
Recall that $p^{k,\pi}_{(f)}$ is the probability of query for the $k^{\text{th}}$ message under the random key $(f,\pi)$, and $p_d^k \triangleq \sum_{\pi \in \mathcal{P}} p^{k,\pi}_{(\underline{0_{K-1}})}$. For simplicity, we will write $p^{k,\pi}_{(\underline{0_{K-1}})}$ as $p^{k,\pi}_{(0)}$ in the sequel.

We first show that the optimal value of the optimization problem $({P1})$ above, which is achieved under the optimal probability distribution in the generalized TSC code, is the same as the optimal value of the optimization problem $({P2})$ below, which is achieved by the optimal distribution allocated to the reduced TSC code. 
\begin{equation}
\begin{aligned}
\min_{p_{0:K-1}} \quad & \frac{N}{N-1} \left( 1-p_0 \right) \\
\textrm{s.t.} \quad 
  & p_j-e^{\epsilon} p_{j+1}\leq 0,\quad \forall j\in [0:K-2],  \\
  & p_j-e^{\epsilon} p_{j-1}\leq 0,\quad \forall j\in [1:K-1], \\
  & p_0, p_1,\ldots,p_{K-1} \geq 0, \\
  & \sum_{j=0}^{K-1}Ns_j p_j=1.
\end{aligned}
\end{equation}

\begin{prop}
\label{prop:P1P2}
    $(P1)=(P2).$
\end{prop}

\begin{proof}
    See Appendix \ref{apx:proof-prop-P1P2}
\end{proof}

\subsection{Optimal Allocation of the Reduced Code}
\label{sec:allocate-reduced-gTSC}

The Lagrangian function of problem $(P2)$ is 
\begin{align}
    \mathscr{L}=& \frac{N}{N-1} (1-p_{0})\  -\sum^{K-1}_{i=0} \mu_{j} p_{j}+\lambda \left( \sum^{K-1}_{j=0} Ns_{j}p_{j}-1\right)  \notag \\
    & +\sum^{K-2}_{i=0} \alpha_{j} \left( p_{j}-e^{\epsilon }p_{j+1}\right)  +\sum^{K-1}_{j=1} \beta_{j} \left( p_{j}-e^{\epsilon }p_{j-1}\right)  
\end{align}

Then the KKT condition can be derived as follows:
\begin{enumerate}
    \item stationarity:
    \begin{align}
        \begin{cases} -\frac{N}{N-1} -\mu_{0} +\lambda Ns_0+\alpha_{0} -e^{\epsilon }\beta_{1}  =0\\ 
        -\mu_{j} +\lambda Ns_{j}+\alpha_{j} - e^{\epsilon } \alpha_{j-1} +\beta_{j} - e^{\epsilon } \beta_{j+1}  = 0, \\
        \qquad\qquad\qquad\qquad\qquad \forall j\in [1:K-2]\\ 
        -\mu_{K-1} +\lambda Ns_{K-1}- e^{\epsilon } \alpha_{K-2} +\beta_{K-1} =\  0 \end{cases} 
    \end{align}

    \item primal feasibility:
    \begin{align}
        \begin{cases}
        N\sum^{K-1}_{j=0} s_{j}p_{j}-1=0\\ 
        p_{j}\geq 0, ~\forall j\in [0:K-1]\\ 
        p_j-e^{\epsilon} p_{j+1} \leq 0 ,~ \forall j\in [0:K-2] \\
        p_j-e^{\epsilon} p_{j-1} \leq 0 ,~ \forall j\in [1:K-1] \\
        \end{cases} 
    \end{align}

    \item dual feasibility:
    \begin{align}
        \begin{cases}\mu_{j } \geq 0 , ~ \forall j\in [0:K-1] \\
        \alpha_{j } \geq 0, \beta_{j+1} \geq 0 ,~ \forall j\in [0:K-2] \\
        \end{cases} 
    \end{align}

    \item complementary slackness:
    \begin{align}
        \begin{cases}\mu_{j } p_{j }=0 ,~ \forall j\in [0:K-1] \\ 
        \alpha_{j }(p_j-e^{\epsilon} p_{j+1}) =0, ~ \forall j\in [0:K-2] \\ 
        \beta_{j}(p_j-e^{\epsilon} p_{j-1})=0,~ \forall j\in [1:K-1] \end{cases} 
    \end{align}
\end{enumerate}

A solution to the KKT conditions is given as follows:
 \begin{enumerate}
    \item primal variables: for $j\in [0:K-1]$,
    \begin{align}
        p_{j} &= \frac{e^{-j\epsilon }}{\sum^{K-1}_{i=0} Ns_{i}e^{-i\epsilon }} = \frac{e^{(K-1-j)\epsilon }}{N\left( e^{\epsilon }+N-1\right)^{K-1} } .
    \end{align}

    \item dual variables:
    \begin{align}
        \begin{cases}
        \mu_j =0 ,~\forall j\in [0:K-1]\\
        \beta_j = 0 ,~\forall j\in [1:K-1]\\
        \lambda = \frac{e^{(K-2)\epsilon }}{\left( N-1\right)  \sum^{K-1}_{i=0} s_{i}e^{\left( K-2-i\right)  \epsilon }} \\
        \alpha_{j\  } =\   \frac{N}{N-1} e^{j\epsilon }\left[ 1-\frac{\sum^{j}_{i=0} s_{i}e^{(K-2-i)\epsilon }}{\sum^{K-1}_{i=0} s_{i}e^{\left( K-2-i\right)  \epsilon }} \right], \\ \qquad\qquad\qquad\qquad \forall j\in [0:K-2] \end{cases}
    \end{align}
\end{enumerate}

It can be verified that the solution above satisfies all KKT conditions. 
The download cost is therefore given by
\begin{align}
    D^{\text{TSC}}\left( \epsilon \right) 
    &=1+\frac{\left( e^{\epsilon }+N-1\right)^{K-1}  -e^{(K-1)\epsilon }}{(N-1)\left( e^{\epsilon }+N-1\right)^{K-1}  } .
\end{align}



\section{Conclusion}

We study the problem of leaky private information retrieval (L-PIR) and optimize the trade-off between communication cost and user privacy leakage under the measurement of pure differential privacy. The optimized probability allocation adopts a simple layered structure: the retrieval using a random key with lower Hamming weight is assigned higher probability. Compared with other state-of-the-art L-PIR schemes, the proposed scheme achieves the leakage ratio exponent ${O}(\log K)$, which is a significant order-wise improvement over ${\Theta}(K)$. A promising direction in the future is the extension to multi-message retrieval.


\bibliographystyle{IEEEtran}
\bibliography{PIR}

\newpage
\appendices

\section{The TSC Code When $K=2$ and $N=3$}
\label{apx:permuteTSC}
We provide another example of the generalized TSC code with permutation in Table \ref{tab:W1} and Table \ref{tab:W2}. We consider a total of $N=3$ servers and $K=2$ messages. The message length is $L=N-1=2$, and we write $W_1 = (a_1,a_2)$, $W_2 = (b_1,b_2)$, where the dummy symbols $a_0$ and $b_0$ are omitted for conciseness. The random key $F^*$ has a total of $18$ possible realizations, each of which is associated with a random vector and the downshifted permutation function $(F,\pi)$. The queries in the TSC code are assigned with different probabilities according to their interference signals and permutation. Note that the interference signal is controlled by the first $(K-1)$ entries $F$ of the random key $F^*$. Let us denote $\|F\|$ as the \emph{size} of the interference corresponding to the random key $F$, which is also its Hamming weight. In this example, $\|F\|$ can only be $0$ or $1$. Note that $|\mathcal{F}|$ is used to denote the cardinality of the set $\mathcal{F}$, different from $\|F\|$.

\begin{table}[tb!]
\centering
\resizebox{\linewidth}{!}{
\renewcommand{\arraystretch}{1.2}
\begin{tabular}{|ccccccccc|}
\hline
\multicolumn{9}{|c|}{Requesting Message $k=1$}                                                                                                                                                                                                                                               \\ \hline
\multicolumn{1}{|c|}{\multirow{2}{*}{Prob.}} & \multicolumn{1}{c|}{\multirow{2}{*}{$F$}} & \multicolumn{1}{c|}{\multirow{2}{*}{$\pi$}}&\multicolumn{2}{c|}{Server 1}                                         & \multicolumn{2}{c|}{Server 2}                                         & \multicolumn{2}{c|}{Server 3}                    \\ \cline{4-9} 
\multicolumn{1}{|c|}{}                       & \multicolumn{1}{c|}{}      & \multicolumn{1}{c|}{}               & \multicolumn{1}{c|}{$Q_{1}^{[1]}$} & \multicolumn{1}{c|}{$A_{1}$}     & \multicolumn{1}{c|}{$Q_{2}^{[1]}$} & \multicolumn{1}{c|}{$A_{2}$}     & \multicolumn{1}{c|}{$Q_{3}^{[1]}$} & $A_{3}$     \\ \hline\hline


\multicolumn{1}{|c|}{$p_{(0)}^{1,{[2,1,0]}}$}    & \multicolumn{1}{c|}{$0$}  & \multicolumn{1}{c|}{$(2,1,0)$}                & \multicolumn{1}{c|}{$20$}          & \multicolumn{1}{c|}{$a_2$}       & \multicolumn{1}{c|}{$10$}          & \multicolumn{1}{c|}{$a_1$}       & \multicolumn{1}{c|}{$00$}          & $\emptyset$ \\ \hline
\multicolumn{1}{|c|}{$p_{(0)}^{1,[2, 0, 1]}$}    & \multicolumn{1}{c|}{$0$}  & \multicolumn{1}{c|}{$(2,0,1)$}                & \multicolumn{1}{c|}{$20$}          & \multicolumn{1}{c|}{$a_2$}       & \multicolumn{1}{c|}{$00$}          & \multicolumn{1}{c|}{$\emptyset$} & \multicolumn{1}{c|}{$10$}          & $a_1$       \\ \hline
\multicolumn{1}{|c|}{$p_{(0)}^{1,[1,2,0]}$}    & \multicolumn{1}{c|}{$0$}    & \multicolumn{1}{c|}{$(1,2,0)$}              & \multicolumn{1}{c|}{$10$}          & \multicolumn{1}{c|}{$a_1$}       & \multicolumn{1}{c|}{$20$}          & \multicolumn{1}{c|}{$a_2$}       & \multicolumn{1}{c|}{$00$}          & $\emptyset$ \\ \hline
\multicolumn{1}{|c|}{$p_{(0)}^{1,[1,0,2]}$}    & \multicolumn{1}{c|}{$0$}    & \multicolumn{1}{c|}{$(1,0,2)$}              & \multicolumn{1}{c|}{$10$}          & \multicolumn{1}{c|}{$a_1$}       & \multicolumn{1}{c|}{$00$}          & \multicolumn{1}{c|}{$\emptyset$} & \multicolumn{1}{c|}{$20$}          & $a_2$       \\ \hline
\multicolumn{1}{|c|}{$p_{(0)}^{1,[0,2,1]}$}    & \multicolumn{1}{c|}{$0$}    & \multicolumn{1}{c|}{$(0,2,1)$}              & \multicolumn{1}{c|}{$00$}          & \multicolumn{1}{c|}{$\emptyset$} & \multicolumn{1}{c|}{$20$}          & \multicolumn{1}{c|}{$a_2$}       & \multicolumn{1}{c|}{$10$}          & $a_1$       \\ \hline
\multicolumn{1}{|c|}{$p_{(0)}^{1,[0,1,2]}$}    & \multicolumn{1}{c|}{$0$}    & \multicolumn{1}{c|}{$(0,1,2)$}              & \multicolumn{1}{c|}{$00$}          & \multicolumn{1}{c|}{$\emptyset$} & \multicolumn{1}{c|}{$10$}          & \multicolumn{1}{c|}{$a_1$}       & \multicolumn{1}{c|}{$20$}          & $a_2$       \\ \hline\hline

\multicolumn{1}{|c|}{$p_{(1)}^{1,[2,1,0]}$}    & \multicolumn{1}{c|}{$1$}    & \multicolumn{1}{c|}{$(2,1,0)$}              & \multicolumn{1}{c|}{$11$}          & \multicolumn{1}{c|}{$a_1\oplus b_1$}   & \multicolumn{1}{c|}{$01$}          & \multicolumn{1}{c|}{$b_1$}       & \multicolumn{1}{c|}{$21$}          & $a_2\oplus b_1$   \\ \hline
\multicolumn{1}{|c|}{$p_{(1)}^{1,[2,0,1]}$}    & \multicolumn{1}{c|}{$1$}    & \multicolumn{1}{c|}{$(2,0,1)$}              & \multicolumn{1}{c|}{$11$}          & \multicolumn{1}{c|}{$a_1\oplus b_1$}   & \multicolumn{1}{c|}{$21$}          & \multicolumn{1}{c|}{$a_2\oplus b_1$}   & \multicolumn{1}{c|}{$01$}          & $b_1$       \\ \hline
\multicolumn{1}{|c|}{$p_{(1)}^{1,[1,2,0]}$}    & \multicolumn{1}{c|}{$1$}    & \multicolumn{1}{c|}{$(1,2,0)$}              & \multicolumn{1}{c|}{$01$}          & \multicolumn{1}{c|}{$b_1$}       & \multicolumn{1}{c|}{$11$}          & \multicolumn{1}{c|}{$a_1\oplus b_1$}   & \multicolumn{1}{c|}{$21$}          & $a_2\oplus b_1$   \\ \hline
\multicolumn{1}{|c|}{$p_{(1)}^{1,[1,0,2]}$}    & \multicolumn{1}{c|}{$1$}    & \multicolumn{1}{c|}{$(1,0,2)$}              & \multicolumn{1}{c|}{$01$}          & \multicolumn{1}{c|}{$b_1$}       & \multicolumn{1}{c|}{$21$}          & \multicolumn{1}{c|}{$a_2\oplus b_1$}   & \multicolumn{1}{c|}{$11$}          & $a_1\oplus b_1$   \\ \hline
\multicolumn{1}{|c|}{$p_{(1)}^{1,[0,2,1]}$}    & \multicolumn{1}{c|}{$1$}    & \multicolumn{1}{c|}{$(0,2,1)$}              & \multicolumn{1}{c|}{$21$}          & \multicolumn{1}{c|}{$a_2\oplus b_1$}   & \multicolumn{1}{c|}{$11$}          & \multicolumn{1}{c|}{$a_1\oplus b_1$}   & \multicolumn{1}{c|}{$01$}          & $b_1$       \\ \hline
\multicolumn{1}{|c|}{$p_{(1)}^{1,[0,1,2]}$}    & \multicolumn{1}{c|}{$1$}    & \multicolumn{1}{c|}{$(0,1,2)$}              & \multicolumn{1}{c|}{$21$}          & \multicolumn{1}{c|}{$a_2\oplus b_1$}   & \multicolumn{1}{c|}{$01$}          & \multicolumn{1}{c|}{$b_1$}       & \multicolumn{1}{c|}{$11$}          & $a_1\oplus b_1$   \\ \hline
\multicolumn{1}{|c|}{$p_{(2)}^{1,[2,1,0]}$}    & \multicolumn{1}{c|}{$2$}    & \multicolumn{1}{c|}{$(2,1,0)$}              & \multicolumn{1}{c|}{$02$}          & \multicolumn{1}{c|}{$b_2$}       & \multicolumn{1}{c|}{$22$}          & \multicolumn{1}{c|}{$a_2\oplus b_2$}   & \multicolumn{1}{c|}{$12$}          & $a_1\oplus b_2$   \\ \hline
\multicolumn{1}{|c|}{$p_{(2)}^{1,[2,0,1]}$}    & \multicolumn{1}{c|}{$2$}    & \multicolumn{1}{c|}{$(2,0,1)$}              & \multicolumn{1}{c|}{$02$}          & \multicolumn{1}{c|}{$b_2$}       & \multicolumn{1}{c|}{$12$}          & \multicolumn{1}{c|}{$a_1\oplus b_2$}   & \multicolumn{1}{c|}{$22$}          & $a_2\oplus b_2$   \\ \hline
\multicolumn{1}{|c|}{$p_{(2)}^{1,[1,2,0]}$}    & \multicolumn{1}{c|}{$2$}    & \multicolumn{1}{c|}{$(1,2,0)$}              & \multicolumn{1}{c|}{$22$}          & \multicolumn{1}{c|}{$a_2\oplus b_2$}   & \multicolumn{1}{c|}{$02$}          & \multicolumn{1}{c|}{$b_2$}       & \multicolumn{1}{c|}{$12$}          & $a_1\oplus b_2$   \\ \hline
\multicolumn{1}{|c|}{$p_{(2)}^{1,[1,0,2]}$}    & \multicolumn{1}{c|}{$2$}    & \multicolumn{1}{c|}{$(1,0,2)$}              & \multicolumn{1}{c|}{$22$}          & \multicolumn{1}{c|}{$a_2\oplus b_2$}   & \multicolumn{1}{c|}{$12$}          & \multicolumn{1}{c|}{$a_1\oplus b_2$}   & \multicolumn{1}{c|}{$02$}          & $b_2$       \\ \hline
\multicolumn{1}{|c|}{$p_{(2)}^{1,[0,2,1]}$}    & \multicolumn{1}{c|}{$2$}    & \multicolumn{1}{c|}{$(0,2,1)$}              & \multicolumn{1}{c|}{$12$}          & \multicolumn{1}{c|}{$a_1\oplus b_2$}   & \multicolumn{1}{c|}{$02$}          & \multicolumn{1}{c|}{$b_2$}       & \multicolumn{1}{c|}{$22$}          & $a_2\oplus b_2$   \\ \hline
\multicolumn{1}{|c|}{$p_{(2)}^{1,[0,1,2]}$}    & \multicolumn{1}{c|}{$2$}    & \multicolumn{1}{c|}{$(0,1,2)$}              & \multicolumn{1}{c|}{$12$}          & \multicolumn{1}{c|}{$a_1\oplus b_2$}   & \multicolumn{1}{c|}{$22$}          & \multicolumn{1}{c|}{$a_2\oplus b_2$}   & \multicolumn{1}{c|}{$02$}          & $b_2$       \\ \hline
\end{tabular}
}
\caption{TSC code for $N=3, K=2$: retrieval of $W_{1}$}
\label{tab:W1}
\end{table}

\begin{table}[tb!]
\centering
\resizebox{\linewidth}{!}{
\renewcommand{\arraystretch}{1.13}
\begin{tabular}{|ccccccccc|}
\hline
\multicolumn{9}{|c|}{Requesting Message $k=2$}                                                                                                                                                                                                                                               \\ \hline
\multicolumn{1}{|c|}{\multirow{2}{*}{Prob.}} & \multicolumn{1}{c|}{\multirow{2}{*}{$F$}} & \multicolumn{1}{c|}{\multirow{2}{*}{$\pi$}}&\multicolumn{2}{c|}{Server 1}                                         & \multicolumn{2}{c|}{Server 2}                                         & \multicolumn{2}{c|}{Server 3}                    \\ \cline{4-9} 
\multicolumn{1}{|c|}{}                       & \multicolumn{1}{c|}{}      & \multicolumn{1}{c|}{}               & \multicolumn{1}{c|}{$Q_{1}^{[2]}$} & \multicolumn{1}{c|}{$A_{1}$}     & \multicolumn{1}{c|}{$Q_{2}^{[2]}$} & \multicolumn{1}{c|}{$A_{2}$}     & \multicolumn{1}{c|}{$Q_{3}^{[2]}$} & $A_{3}$     \\ \hline\hline


\multicolumn{1}{|c|}{$p_{(0)}^{2,{[2,1,0]}}$}    & \multicolumn{1}{c|}{$0$}  & \multicolumn{1}{c|}{$(2,1,0)$}                & \multicolumn{1}{c|}{$02$}          & \multicolumn{1}{c|}{$b_2$}       & \multicolumn{1}{c|}{$01$}          & \multicolumn{1}{c|}{$b_1$}       & \multicolumn{1}{c|}{$00$}          & $\emptyset$ \\ \hline
\multicolumn{1}{|c|}{$p_{(0)}^{2,[2, 0, 1]}$}    & \multicolumn{1}{c|}{$0$}  & \multicolumn{1}{c|}{$(2,0,1)$}                & \multicolumn{1}{c|}{$02$}          & \multicolumn{1}{c|}{$b_2$}       & \multicolumn{1}{c|}{$00$}          & \multicolumn{1}{c|}{$\emptyset$} & \multicolumn{1}{c|}{$01$}          & $b_1$       \\ \hline
\multicolumn{1}{|c|}{$p_{(0)}^{2,[1,2,0]}$}    & \multicolumn{1}{c|}{$0$}    & \multicolumn{1}{c|}{$(1,2,0)$}              & \multicolumn{1}{c|}{$01$}          & \multicolumn{1}{c|}{$b_1$}       & \multicolumn{1}{c|}{$02$}          & \multicolumn{1}{c|}{$b_2$}       & \multicolumn{1}{c|}{$00$}          & $\emptyset$ \\ \hline
\multicolumn{1}{|c|}{$p_{(0)}^{2,[1,0,2]}$}    & \multicolumn{1}{c|}{$0$}    & \multicolumn{1}{c|}{$(1,0,2)$}              & \multicolumn{1}{c|}{$01$}          & \multicolumn{1}{c|}{$b_1$}       & \multicolumn{1}{c|}{$00$}          & \multicolumn{1}{c|}{$\emptyset$} & \multicolumn{1}{c|}{$02$}          & $b_2$       \\ \hline
\multicolumn{1}{|c|}{$p_{(0)}^{2,[0,2,1]}$}    & \multicolumn{1}{c|}{$0$}    & \multicolumn{1}{c|}{$(0,2,1)$}              & \multicolumn{1}{c|}{$00$}          & \multicolumn{1}{c|}{$\emptyset$} & \multicolumn{1}{c|}{$02$}          & \multicolumn{1}{c|}{$b_2$}       & \multicolumn{1}{c|}{$01$}          & $b_1$       \\ \hline
\multicolumn{1}{|c|}{$p_{(0)}^{2,[0,1,2]}$}    & \multicolumn{1}{c|}{$0$}    & \multicolumn{1}{c|}{$(0,1,2)$}              & \multicolumn{1}{c|}{$00$}          & \multicolumn{1}{c|}{$\emptyset$} & \multicolumn{1}{c|}{$01$}          & \multicolumn{1}{c|}{$b_1$}       & \multicolumn{1}{c|}{$02$}          & $b_2$       \\ \hline \hline

\multicolumn{1}{|c|}{$p_{(1)}^{2,[2,1,0]}$}    & \multicolumn{1}{c|}{$1$}    & \multicolumn{1}{c|}{$(2,1,0)$}              & \multicolumn{1}{c|}{$11$}          & \multicolumn{1}{c|}{$a_1\oplus b_1$}   & \multicolumn{1}{c|}{$10$}          & \multicolumn{1}{c|}{$a_1$}       & \multicolumn{1}{c|}{$12$}          & $a_1\oplus b_2$   \\ \hline
\multicolumn{1}{|c|}{$p_{(1)}^{2,[2,0,1]}$}    & \multicolumn{1}{c|}{$1$}    & \multicolumn{1}{c|}{$(2,0,1)$}              & \multicolumn{1}{c|}{$11$}          & \multicolumn{1}{c|}{$a_1\oplus b_1$}   & \multicolumn{1}{c|}{$12$}          & \multicolumn{1}{c|}{$a_1\oplus b_2$}   & \multicolumn{1}{c|}{$10$}          & $a_1$       \\ \hline
\multicolumn{1}{|c|}{$p_{(1)}^{2,[1,2,0]}$}    & \multicolumn{1}{c|}{$1$}    & \multicolumn{1}{c|}{$(1,2,0)$}              & \multicolumn{1}{c|}{$10$}          & \multicolumn{1}{c|}{$a_1$}       & \multicolumn{1}{c|}{$11$}          & \multicolumn{1}{c|}{$a_1\oplus b_1$}   & \multicolumn{1}{c|}{$12$}          & $a_1\oplus b_2$   \\ \hline
\multicolumn{1}{|c|}{$p_{(1)}^{2,[1,0,2]}$}    & \multicolumn{1}{c|}{$1$}    & \multicolumn{1}{c|}{$(1,0,2)$}              & \multicolumn{1}{c|}{$10$}          & \multicolumn{1}{c|}{$a_1$}       & \multicolumn{1}{c|}{$12$}          & \multicolumn{1}{c|}{$a_1\oplus b_2$}   & \multicolumn{1}{c|}{$11$}          & $a_1\oplus b_1$   \\ \hline
\multicolumn{1}{|c|}{$p_{(1)}^{2,[0,2,1]}$}    & \multicolumn{1}{c|}{$1$}    & \multicolumn{1}{c|}{$(0,2,1)$}              & \multicolumn{1}{c|}{$12$}          & \multicolumn{1}{c|}{$a_1\oplus b_2$}   & \multicolumn{1}{c|}{$11$}          & \multicolumn{1}{c|}{$a_1\oplus b_1$}   & \multicolumn{1}{c|}{$10$}          & $a_1$       \\ \hline
\multicolumn{1}{|c|}{$p_{(1)}^{2,[0,1,2]}$}    & \multicolumn{1}{c|}{$1$}    & \multicolumn{1}{c|}{$(0,1,2)$}              & \multicolumn{1}{c|}{$12$}          & \multicolumn{1}{c|}{$a_1\oplus b_2$}   & \multicolumn{1}{c|}{$10$}          & \multicolumn{1}{c|}{$a_1$}       & \multicolumn{1}{c|}{$11$}          & $a_1\oplus b_1$   \\ \hline
\multicolumn{1}{|c|}{$p_{(2)}^{2,[2,1,0]}$}    & \multicolumn{1}{c|}{$2$}    & \multicolumn{1}{c|}{$(2,1,0)$}              & \multicolumn{1}{c|}{$20$}          & \multicolumn{1}{c|}{$a_2$}       & \multicolumn{1}{c|}{$22$}          & \multicolumn{1}{c|}{$a_2\oplus b_2$}   & \multicolumn{1}{c|}{$21$}          & $a_2\oplus b_1$   \\ \hline
\multicolumn{1}{|c|}{$p_{(2)}^{2,[2,0,1]}$}    & \multicolumn{1}{c|}{$2$}    & \multicolumn{1}{c|}{$(2,0,1)$}              & \multicolumn{1}{c|}{$20$}          & \multicolumn{1}{c|}{$a_2$}       & \multicolumn{1}{c|}{$21$}          & \multicolumn{1}{c|}{$a_2\oplus b_1$}   & \multicolumn{1}{c|}{$22$}          & $a_2\oplus b_2$   \\ \hline
\multicolumn{1}{|c|}{$p_{(2)}^{2,[1,2,0]}$}    & \multicolumn{1}{c|}{$2$}    & \multicolumn{1}{c|}{$(1,2,0)$}              & \multicolumn{1}{c|}{$22$}          & \multicolumn{1}{c|}{$a_2\oplus b_2$}   & \multicolumn{1}{c|}{$20$}          & \multicolumn{1}{c|}{$a_2$}       & \multicolumn{1}{c|}{$21$}          & $a_2\oplus b_1$   \\ \hline
\multicolumn{1}{|c|}{$p_{(2)}^{2,[1,0,2]}$}    & \multicolumn{1}{c|}{$2$}    & \multicolumn{1}{c|}{$(1,0,2)$}              & \multicolumn{1}{c|}{$22$}          & \multicolumn{1}{c|}{$a_2\oplus b_2$}   & \multicolumn{1}{c|}{$21$}          & \multicolumn{1}{c|}{$a_2\oplus b_1$}   & \multicolumn{1}{c|}{$20$}          & $a_2$       \\ \hline
\multicolumn{1}{|c|}{$p_{(2)}^{2,[0,2,1]}$}    & \multicolumn{1}{c|}{$2$}    & \multicolumn{1}{c|}{$(0,2,1)$}              & \multicolumn{1}{c|}{$21$}          & \multicolumn{1}{c|}{$a_2\oplus b_1$}   & \multicolumn{1}{c|}{$20$}          & \multicolumn{1}{c|}{$a_2$}       & \multicolumn{1}{c|}{$22$}          & $a_2\oplus b_2$   \\ \hline
\multicolumn{1}{|c|}{$p_{(2)}^{2,[0,1,2]}$}    & \multicolumn{1}{c|}{$2$}    & \multicolumn{1}{c|}{$(0,1,2)$}              & \multicolumn{1}{c|}{$21$}          & \multicolumn{1}{c|}{$a_2\oplus b_1$}   & \multicolumn{1}{c|}{$22$}          & \multicolumn{1}{c|}{$a_2\oplus b_2$}   & \multicolumn{1}{c|}{$20$}          & $a_2$       \\ \hline
\end{tabular}}
\caption{TSC code for $N=3, K=2$: retrieval of $W_{2}$}
\label{tab:W2}
\end{table}

\section{Proof of Proposition \ref{prop:P1P2}}
\label{apx:proof-prop-P1P2}

\begin{proof}
The direction $(P1)\leq (P2)$: this direction is trivially true since $(P2)$ can be viewed as $(P1)$ under the additional constraints enforced through \eqref{eqn:reduced}. 

The direction $(P1)\geq (P2)$: 
Recall that the query $Q_n^{[k]}$ can take any possible values in $\mathcal{Q}$. 
Denote $t_j\triangleq |\{q \in \mathcal{Q}:\|q\|=j\}|$, which is calculated as 
\begin{align}
    t_j= \left(\begin{array}{c}K\\j\end{array}\right)(N-1)^j, \forall~j\in[0:K].   \label{coefficient-t-j}
\end{align}
For notational simplicity, let $p_{-1}=p_K=0$. Similarly, we use $s_j$ to denote $|\mathcal{F}_j|$, i.e, the number of random key $f$ that has Hamming weight $j$, given by
\begin{align}
    s_j= \left(\begin{array}{c}K-1\\j\end{array}\right)(N-1)^j, \forall~j\in[0:K-1].    \label{coefficient-s-j}
\end{align}

Given an optimal solution in $(P1)$, we can find the following assignment of $p_0, \ldots, p_{K-1}$: 
\begin{align}
    p_j =\frac{1}{NKs_j}\sum_{k=1}^{K}\sum_{\pi\in \Pc}\sum_{f\in\Fc_j} p_{(f)}^{k,\pi}, \quad j\in [0:K-1].   \label{def-p-j-proof}
\end{align}

With the relation \eqref{def-p-j-proof}, we have $p_j \geq 0$ for any $j \in [0:K-1]$, and moreover, 
\begin{align}
   &\sum_{j=0}^{K-1} Ns_jp_j = \frac{1}{K} \sum_{k=1}^{K}\left(\sum_{\pi\in \Pc}\sum_{f\in\Fc}p_{(f)}^{k,\pi}\right)=1, \label{prob-constraint-proof}
\end{align}

For the $\epsilon$-DP constraints, it is trivially true when $\|q\|=0$ given the fact that for all $k\in[1:K]$,
\begin{equation}
\begin{aligned}
    &\frac{1}{NK}\sum_{k=1}^{K}\sum_{n=1}^{N}\sum_{\pi: \phi^*_n(k,(f,\pi))=0}p^{k,\pi}_{(f)} \\
    = &\frac{1}{NK}\sum_{k=1}^{K}\sum_{\pi}p_{(0)}^{k,\pi} \\
    = &p_0 .
\end{aligned}
\end{equation}

When $\|q\|=K$, we first introduce the notation
\begin{align}
     & q|k \triangleq (q_1,q_2,\ldots,q_{k-1},q_{k+1},\ldots,q_{K}), \\
     & q_{n,k} \triangleq Q^{[k]}_{n,k},
\end{align}
which denotes the query vector with the $k$-th symbol removed and the $k$-th symbol of $Q_{n}^{[k]}$, respectively. Then we have
\begin{align}
    \sum_{n=1}^{N} \sum_{(f,\pi):\phi_{n}^*(k,(f,\pi))=q}p_{(f)}^{k,\pi}= \sum_{\pi}p_{(q|k)}^{k,\pi},
\end{align}
since for that fixed $q$, the corresponding $f$ is fixed, and for each $\pi$, there is one and only one $n$ such that $\phi_{n}^*(k,(f,\pi))=q$ holds. 

Therefore, 
\begin{equation}
\begin{aligned}
    & \sum_{\|q\|=K} \sum_{n=1}^{N} \sum_{(f,\pi):\phi_{n}^*(k,(f,\pi))=q}p_{(f)}^{k,\pi} \\
    = & \sum_{\|q\|=K}\sum_{\pi}p_{(q|k)}^{k,\pi} \\
    = &(N-1)\sum_{\|f\|=K-1}\sum_{\pi}p_{(f)}^{k,\pi},
\end{aligned}
\end{equation}
since for each $k$, each $f$ with $\|f\|=K-1$ corresponds to exactly $N-1$ queries with $\|q\|=K$. Therefore, for all $k\in[1:k]$,
\begin{align}
    &\frac{1}{NK} \sum_{k=1}^{K}\sum_{\|q\|=K}\sum_{n=1}^{N}\sum_{\pi: \phi^*_n(k,(f,\pi))=q}p^{k,\pi}_{(f)} \notag \\
   =&\frac{N-1}{NK}\sum_{k=1}^{K}\sum_{\|f\|=K-1}\sum_{\pi}p_{(f)}^{k,\pi} \\
   =&t_K p_{K-1},
\end{align}
which implies that the $\epsilon$-DP constraint is true for $\|q\|=K$.

When $\|q\|=j\in[1:K-1]$, in a similar manner, we observe that:
\begin{align}
    &\frac{1}{KN} \sum_{k=1}^{K}\sum_{\|q\|=j}\sum_{n=1}^{N}\sum_{\pi: \phi^*_n(k,(f,\pi))=q}p^{k,\pi}_{(f)} \\
    = &\frac{1}{KN} \sum_{k=1}^{K} \sum_{\|q\|=j}\sum_{\pi}p_{(q|k)}^{k,\pi} \\
    = & \frac{1}{N} \sum_{\| q\| =j} \left\{ \frac{j}{K}\left( \frac{1}{j} \sum^{K}_{k=1} \mathds{1}(\| q|k\| =j-1)\sum_{\pi } p^{k,\pi }_{(q|k)}\right) \right. \notag \\
    & \left. + \frac{K-j}{K} \left( \frac{1}{K-j} \sum^{K}_{k=1} \mathds{1}(\| q|k\| =j)\sum_{\pi } p^{k,\pi }_{(q|k)}\right) \right\}  \label{eqn:countj} \\
    = & \frac{j}{K}\left( \frac{1}{N} \frac{N-1}{j} \sum^{K}_{k=1} \sum_{\| f\| =j-1} \sum_{\pi } p^{k,\pi }_{(f)}\right) \notag \\
    & + \frac{K-j}{K} \left( \frac{1}{N} \frac{1}{K-j} \sum^{K}_{k=1} \sum_{\| f\| =j} \sum_{\pi } p^{k,\pi }_{(f)}\right)  \label{permute-reduce-9-4}\\
    = & \frac{j}{K}\left( \frac{1}{N} \frac{t_{j}}{Ks_{j-1}} \sum^{K}_{k=1} \sum_{\| f\| =j-1} \sum_{\pi } p^{k,\pi }_{(f)}\right) \notag \\ 
    & + \frac{K-j}{K} \left(\frac{1}{N} \frac{t_{j}}{Ks_{j}} \sum^{K}_{k=1} \sum_{\| f\| =j} \sum_{\pi } p^{k,\pi }_{(f)}\right) \label{permute-reduce-9-5}\\
    = & \frac{j}{K} t_j p_{j-1} + \frac{K-j}{K} t_j p_{j}. \label{permute-reduce-9-6}
\end{align}
where \eqref{eqn:countj} follows from $\left\{ f:f=q|k,\| f\| =j-1\right\}  =j$ and $\left\{ f:f=q|k,\| f\| =j\right\} =K-j$ for a given $q$, and $\mathds{1}(\cdot)$ is the indicator function, \eqref{permute-reduce-9-4} follows by counting the number of $f$ for each $q$, 
\eqref{permute-reduce-9-5} follows from $\frac{t_j}{s_{j-1}}=\frac{K(N-1)}{j}$ and $\frac{t_j}{s_j}=\frac{K}{K-j}$, and \eqref{permute-reduce-9-6} follows by the assignment of $p_j$ in \eqref{def-p-j-proof}. Therefore, we have the following bounds for every $j\in[1:K-1]$:
\begin{align}
& t_j \cdot \min \{p_{j-1}, p_j\} \leq \frac{j}{K} t_{j}p_{j-1}+\frac{K-j}{K} t_{j}p_{j} \\
&t_j \cdot \max\{p_{j-1}, p_j\} \geq \frac{j}{K} t_{j}p_{j-1}+\frac{K-j}{K} t_{j}p_{j}.
\end{align}

By $\epsilon$-DP constraint, for every $j\in[1:K-1]$ we need 
\begin{align}
    t_j \cdot \max\{p_{j-1}, p_j\} - t_j  e^{\epsilon} \cdot \min \{p_{j-1},p_j\} \leq 0.
\end{align}
This implies that for all $j\in[1:K-1]$,
\begin{align}
    \begin{cases}p_{j-1}-e^{\epsilon }p_{j}\leq 0,\\ 
    p_{j}-e^{\epsilon }p_{j-1}\leq 0.\end{cases}
\end{align}

It remains to show that this assignment leads to a lower objective function value in $(P2)$ than the optimal value of $(P1)$. By the convexity of the max function,
\begin{align}
    D &=\max_{k}\left\{ \left( \sum_{\pi\in \Pc} p^{k,\pi}_{(0)} \right)+\frac{N}{N-1} \left(1-\sum_{\pi\in \Pc} p^{k,\pi}_{(0)}\right)\right\} \\
    &\geq \frac{1}{K} \sum^{K}_{k=1} \left( \sum_{\pi \in \mathcal{P} } p^{k,\pi }_{(0)}+\frac{N}{N-1} \left( 1-\sum_{\pi \in \mathcal{P} } p^{k,\pi }_{(0)}\right)  \right) \\
    &= Np_{0}+\frac{N}{N-1} \left( 1-Np_{0}\right) \\
    &= \frac{N}{N-1} \left( 1-p_{0}\right) .
\end{align}
This proves the inequality $(P1)\geq (P2)$, and the proof is complete.
\end{proof}

\end{document}